\def\year{2018}\relax
\documentclass[letterpaper,backref=page]{article} 
\usepackage{aaai20}  
\usepackage{times}  
\usepackage{helvet} 
\usepackage{courier}  
\usepackage[hyphens]{url}  
\usepackage{graphicx} 
\usepackage{graphicx}  
\usepackage{amsmath}
\usepackage{amssymb}
\usepackage{amsthm}
\usepackage{graphics}
\usepackage{tikz}
\usepackage{multirow}
\usepackage{siunitx}
\usepackage{booktabs}
\usepackage[toc,page]{appendix} 
\usepackage[ruled,norelsize,linesnumbered]{algorithm2e}
\usepackage{relsize}
\newtheorem{theorem}{Theorem}
\newtheorem{lemma}{Lemma}
\newtheorem{remark}{Remark}
\newtheorem{example}{Example}
\newtheorem{problem}{Problem}
\newtheorem{assumption}{Assumption}

\usepackage{hyperref}
\definecolor{mygreen}{rgb}{0.0, 0.7, 0.05}
\definecolor{myblue}{rgb}{0.0, 0.0, 0.61}
\hypersetup{
 		colorlinks  = true, 
 		urlcolor    = myblue, 
 		linkcolor   = myblue, 
 		citecolor   = myblue, 
 }

\title{Deep Learning--powered Iterative Combinatorial Auctions\footnotemark[1]}
\author{Jakob Weissteiner\\
        University of Zurich \& ETH AI Center\\  
        weissteiner@ifi.uzh.ch
        \And Sven Seuken\\
        University of Zurich \& ETH AI Center\\  
        seuken@ifi.uzh.ch
}
\newcommand{\cited}[1]{\citeauthor{#1}\ \shortcite{#1}}
\DeclareMathOperator*{\argmax}{arg\,\max}
\SetAlFnt{\small}
\makeatletter
\g@addto@macro{\@algocf@init}{\SetKwInOut{Parameter}{Parameter}}
\makeatother
\newenvironment{rcases}{\left.\begin{aligned}}{\end{aligned}\right\rbrace}

\begin{document}
\maketitle
\begin{abstract}
In this paper, we study the design of deep learning-powered iterative combinatorial auctions (ICAs). We build on prior work where preference elicitation was done via kernelized support vector regressions (SVRs). However, the SVR-based approach has limitations because it requires solving a machine learning (ML)-based winner determination problem (WDP). With expressive kernels  (like gaussians), the ML-based WDP cannot be solved for large domains. While linear or quadratic kernels have better computational scalability, these kernels have limited expressiveness. In this work, we address these shortcomings by using deep neural networks (DNNs) instead of SVRs. We first show how the DNN-based WDP can be reformulated into a mixed integer program (MIP). Second, we experimentally compare the prediction performance of DNNs against SVRs.  Third, we present experimental evaluations in two medium-sized domains which show that even ICAs based on relatively small-sized DNNs lead to higher economic efficiency than ICAs based on kernelized SVRs. Finally, we show that our DNN-powered ICA also scales well to very large CA domains.
\end{abstract}

\renewcommand{\thefootnote}{\fnsymbol{footnote}}
\footnotetext[1]{This paper is the slightly updated version of \cited{weissteiner2020deep} published at AAAI'20 including the appendix.}
\renewcommand{\thefootnote}{\arabic{footnote}}

\section{Introduction}\label{Introduction}
Combinatorial auctions (CAs) are used to allocate multiple heterogeneous items to  bidders in domains where these items may be substitutes or complements. Specifically, in a CA, bidders are allowed to submit bids on \emph{bundles} of items rather than on individual items. CAs are widely used in practice, including for the sale of airport landing and take-off slots \cite{rassenti1982combinatorial}, in industrial procurement \cite{bichler2006procurement}, and for the sale of spectrum licenses \cite{cramton2013spectrumauctions}.

One of the main challenges in large CAs is that the bundle space grows exponentially in the number of items. This  typically makes it impossible for the bidders to report their full value function, even for medium-sized domains. Thus, careful preference elicitation is needed in CAs.

\cited{nisan2006communication} have shown that to achieve full efficiency and support general value functions, exponential communication in the number of items is needed in the worst case. Thus, practical auction designs cannot provide efficiency guarantees in large CA domains. Instead, many recent proposals for CAs have focused on \emph{iterative combinatorial auctions (ICAs)} where  the auctioneer interacts with bidders over multiple rounds, eliciting a \textit{limited} amount of information, aiming to find a highly efficient allocation.

ICAs have found widespread application in practice. For example, just between 2008 and 2014, the combinatorial clock auction (CCA) \cite{ausubel2006clock} has been used to conduct more than $15$ spectrum auctions and has generated more than \$20 Billion in total revenue \cite{ausubel2017practical}. Another important application of ICAs are auctions for building offshore wind farms \cite{ausubel2011auction}. Given the value of the resources allocated in these real-world ICAs, increasing their efficiency by 1-2\% points already translates into welfare gains of millions of dollars. Therefore, improving the efficiency of ICAs is an important research challenge.

\subsection{Machine Learning and Mechanism Design}
Researchers have proposed various ways to further increase the efficiency of CAs by integrating machine learning (ML) methods into the mechanism. This research goes back to \cited{blum2004preference} and \cited{lahaie2004applying}, who studied the relationship between computational learning theory and preference elicitation in CAs. More recently, \cited{brero2018bayesian} and \cited{brero2019fast} introduced a Bayesian CA where they integrated ML into a CA to achieve faster convergence. In a different strand of research, \citeauthor{dutting2015payment}\ \shortcite{dutting2015payment,dutting2017optimal}, \citeauthor{narasimhan2016automated}\ \shortcite{narasimhan2016automated} and \citeauthor{golowich2018deep}\ \shortcite{golowich2018deep} used ML to directly learn a new mechanism (following the automated mechanism design paradigm).

Most related to the present paper is the work by \citeauthor{brero2017probably}\ \shortcite{brero2017probably,brero2018combinatorial,brero2019workingpaper}, who proposed an \emph{ML-powered ICA}. The core of their auction is an \emph{ML-powered preference elicitation algorithm}. As part of their  algorithm, they used kernelized support vector regressions (SVRs) to learn the nonlinear value functions of bidders. Recently, \cited{brero2019workingpaper} showed that their ML-based ICA achieves even higher efficiency than the CCA. 
However, because of runtime complexity issues, \citeauthor{brero2018combinatorial} \shortcite{brero2018combinatorial,brero2019workingpaper} focused on SVRs with linear and quadratic kernels. This leaves room for improvement, since bidders' valuations can have more complex structures than can be captured by linear or quadratic kernels.

\subsection{Our Approach Using Deep Neural Networks}
In this paper, we propose using DNNs instead of SVRs in ML-powered ICAs. In each round of the auction, we approximate bidders' value functions by DNNs and subsequently solve an optimization problem, a \textit{DNN-based winner determination problem (WDP)} , to determine which query to ask every bidder in the next round. Since our design involves doing this in each round of the auction, a central requirement for the practical implementation of the auction mechanism is to efficiently solve these DNN-based WDPs. Therefore, we show how to reformulate the WDP based on DNNs with rectified linear units (ReLUs) as activation functions into a (linear) mixed integer program (MIP) (Section \ref{MIPFormulation}).

Our approach is related to a recent line of research that uses MIP formulations to study specific properties of DNNs. For example, \cited{cheng2017maximum} studied resilience properties of DNNs. Similarly, \cited{Fischetti2018} used a MIP formulation for finding adversarial examples in image recognition.

To experimentally evaluate the performance of our DNN-based approach, we use the Spectrum Auction Test Suite (SATS) \cite{weiss2017sats} to generate synthetic CA instances (Section \ref{ExperimentalDesign}). We first compare the prediction performance of DNNs against SVRs in the two medium-sized domains GSVM and LSVM. Then we compare the economic efficiency of our DNN-powered ICA
against the SVR-powered ICA. In GSVM (a domain perfectly suited for the quadratic kernel), our DNN-powered ICA matches the efficiency of the SVR-powered ICA, while in the more complex LSVM domain, our DNN-powered ICA outperforms the SVR-powered ICA by 1.74\% points. 
Finally, we also demonstrate that our DNN-based approach scales well to a very large domain, by evaluating it in the MRVM domain (with 98 items and 10 bidders).
Overall, our results show that, perhaps surprisingly, even small-sized neural networks can be advantageous for the design of ICAs.
\section{Preliminaries}
We now present our formal model and review the ML-powered ICA  by \cited{brero2018combinatorial}.\footnote{We compare our DNN-powered ICA against the mechanism described in \cite{brero2018combinatorial} because, when we wrote this paper, \cite{brero2019workingpaper} was not available yet.  We slightly adopt the notation and use $B_i$ instead of $\widehat{\vartheta}_i.$}
\subsection{Iterative Combinatorial Auction}
We consider a CA setting with  $n$   bidders and $m$ indivisible items. Let $N:=\{1,\ldots,n\}$ and $M:=\{1,\ldots,m\}$ denote the set of bidders and items, respectively. We denote by $x\in \mathcal{X}:=\{0,1\}^m$ a bundle of items represented as an indicator vector, where $x_{j}=1$ iff item $j \in M$ is contained in $x$. Bidders' true preferences over bundles are represented by their (private) value functions $v_i: \{0,1\}^m\to \mathbb{R}_+,\,\, i \in N$, i.e.,  $v_i(x)$ represents bidder $i$'s true value for bundle $x$. Let $v:=(v_1,\ldots,v_n)$ denote the vector of bidders' value functions. The (possibly untruthful) reported valuations are denoted by $\hat{v}_i$ and $\hat{v}$, respectively.

By $a:=(a_1,\ldots,a_n) \in \mathcal{X}^n$ we denote an allocation of bundles to bidders, where $a_i\in \mathcal{X}$ is the bundle bidder $i$ obtains. An allocation $a$ is \textit{feasible} if each item is allocated to at most one bidder, i.e., $\forall j \in M: \sum_{i \in N}a_{ij} \le 1$. We denote the set of feasible allocations by $\mathcal{F}:=\left\{a \in \mathcal{X}^n:\sum_{i \in N}a_{ij} \le 1, \,\,\forall j \in M\right\}$. Payments are denoted by $p = (p_1,\ldots,p_n) \in \mathbb{R}^n$, where  $p_i$ is bidder $i$'s payment. Furthermore, we assume that bidders have quasilinear utility functions  $u_i(a):=v_i(a_i)-p_i$. The (true) \textit{social welfare} of an allocation $a$ is defined as $V(a):=\sum_{i \in N} v_i(a_i).$ Let $a^* \in \argmax_{a \in {\mathcal F}}V(a)$ be a feasible, social-welfare maximizing, i.e., \textit{efficient}, allocation given true value functions $v$. Then the efficiency of any feasible allocation $a \in \mathcal{F}$ is measured in terms of $a^*$ by $\frac{V(a)}{V(a^*)}.$

An ICA \emph{mechanism} defines how the bidders interact with the auctioneer, how the final allocation is determined, and how payments are computed. In this paper, we only consider ICAs that ask bidders to iteratively report their valuations $\hat{v}_i(x)$ for particular bundles $x$ selected by the mechanism. A finite set of such reported bundle-value pairs of bidder $i$ is denoted as ${B_i:=\left\{\left(x^{(k)},\hat{v}_i(x^{(k)})\right)\right\}_{k\in\{1,\ldots,n_i\}},\,n_i\in \mathbb{N},\,x^{(k)}\in \mathcal{X}}$, where $n_i$ is the total number of bundle-value pairs reported by bidder $i$. We let $B:=(B_1,\ldots,B_n)$ denote the tuple of reported bundle-value pairs obtained from all bidders. We define the \textit{reported social welfare} of an allocation $a$ given $B$ as
\begin{align}
\widehat{V}(a|B):=\sum_{i \in N:\, \left(a_i,\hat{v}_i(a_i)\right)\in B_i} \hat{v}_i(a_i),
\end{align} 
where the condition $\left(a_i,\hat{v}_i(a_i)\right)\in B_i$ ensures that only values for reported bundles contribute to the sum. Finally, the optimal feasible allocation  $a^*_{B}$ given $B$ is defined as 
\begin{align}
a^*_{B} \in \argmax_{a \in {\mathcal F}}\widehat{V}(a|B).
\end{align}
In the ICA mechanisms we consider in this paper, the final outcome is only computed based on the reported values $B$ at termination. Specifically, the mechanism determines a feasible allocation $a^*_{B}\in \mathcal{F}$ and charges payments $p$.

As the auctioneer can generally only ask each bidder $i$ for a limited number of bundle-value pairs $B_i$, the ICA mechanism needs a sophisticated preference elicitation algorithm. This leads to the following preference elicitation problem, where the goal is to find an (approximately) \textit{efficient} allocation with a limited number of value queries. More formally: 
\begin{problem}[\textsc{Preference Elicitation in ICA}]\label{PEproblem}
Given a cap $c_e$ on the number of value queries in an ICA, elicit from each bidder $i\in N$ a set of reported bundle-value pairs $B_i$ with $|B_i| \leq c_e$ such that the resulting efficiency
of $a^*_{B}$ is maximized, i.e.,
\begin{align}
B \in \argmax_{B: |B_i| \leq c_e} \dfrac{V(a^*_{B})}{V\left(a^*\right)}.
\end{align}
\end{problem}
In practice, a small domain-dependent cap on the number of queries is chosen, e.g., $ c_e \leq 500$.
\subsection{SVR-powered ICA}
We now present a brief review of the ML-based ICA introduced by \cited{brero2018combinatorial}.
At the core of their auction is an \textit{ML-based preference elicitation algorithm} which we reprint here 
as Algorithm \ref{ea}.

\begin{algorithm}
        \Parameter{Machine learning algorithm $\mathcal A$}
    \SetEndCharOfAlgoLine{}
    \SetKwRepeat{Do}{do}{while}
    $B^0=$ initial tuple of reported bundle-value pairs at $t = 0$  \\
    \Do{$\exists i \in N: a^{(t)}_{i} \notin  B^{t-1}_i$} {
            $t \leftarrow t+1$\\
            \textit{Estimation step:} $\tilde{V}^t := \mathcal A(B^{t-1})$\\
            \textit{Optimization step:} $a^{(t)} \in \argmax\limits_{ a \in {\mathcal F}} \tilde{V}^t(a)$\\
            \For{\emph{each bidder} $i$}{
                    \eIf{$a^{(t)}_{i} \notin  B^{t-1}_i$}
                    {Query value $\hat{v}_i (a^{(t)}_{i})$\\
                            $ B_i^{t}= B_i^{t-1} \cup \left\{\left(a^{(t)}_{i}, \hat{v}_i(a^{(t)}_{i})\right)\right\}$} {$ B_i^{t}= B_i^{t-1}$}
            }
    }
    Output tuple of reported bundle-value pairs $ B^t$
    \caption{\small \textsc{ML-based Elicitation}\, {\scriptsize (Brero et al. 2018)}}
    \label{ea}
\end{algorithm}

This algorithm is a procedure to determine $B$, i.e., for each bidder $i$ a set of reported bundle-value pairs $B_i$. Note that the algorithm is described in terms of a generic ML algorithm $\mathcal{A}$ which is used in the \textit{estimation step} (Line 4) to obtain the estimated social welfare function $\tilde{V}^t$ in iteration $t$. In the \textit{optimization step} (Line 5), an ML-based \textit{winner determination problem} is then solved to find an allocation $a^{(t)}$ that maximizes $\tilde{V}^t$. Finally, given the allocation $a^{(t)}$ from iteration $t$, each bidder $i$ is asked to report his value for the bundle $a^{(t)}_{i}$. The algorithm stops when it reaches an allocation $a^{(t)}$ for which all bidders have already reported their values for the corresponding bundles  $a_{i}^{(t)}$.

As the ML-algorithm $\mathcal{A}$, \cited{brero2018combinatorial}  used a sum of kernelized SVRs, i.e,
\begin{align}
\mathcal{A}(B^{t-1}):=\sum_{i \in N} \text{SVR}_i.
\end{align}
Given a bundle $x$, each $\text{SVR}_i$ computes the predicted value as $\text{SVR}_i(x) = w_i\cdot\phi (x)$, where the weights $w_i$ are determined through training on the reported  bundle-value pairs $B^{t-1}_i$. Kernelized SVRs are a popular non-linear regression technique, where a linear model is fitted on transformed data. The transformation of bundles $x$ is implicitly conducted by setting a kernel $k(x,x'):=\phi(x)^T\phi(x')$ in the dual optimization problem \cite{smola2004tutorial}.

\cited{brero2018combinatorial} called their entire auction mechanism the \textit{Pseudo Vickrey-Clarke-Groves} mechanism (PVM). We reprint it here as Algorithm \ref{pvm}. PVM calls the preference elicitation algorithm (Algorithm \ref{ea}) $n+1$ times: once including all bidders (called the \textit{main economy}) and $n$ times excluding a different bidder in each run (called the \textit{marginal economies}). The motivation for this design, which is inspired by the VCG mechanism, is to obtain payments such that the auction aligns bidders' incentives with allocative efficiency. Here, $B^{(-i)}$ denotes the output of Algorithm \ref{ea} by excluding bidder $i$ from the set of bidders. For each of the reported bundle-value pairs $B^{(-i)}$ obtained from the $n+1$ runs, PVM calculates a corresponding allocation that maximizes the \textit{reported social welfare}  (Line 2). The final allocation $a^{pvm}$ is determined as the allocation of the $n+1$ runs with the largest \textit{reported social welfare} (Line 3). Finally, VCG-style payments are calculated (Line 4). 

\begin{algorithm}[t]
    Run Algorithm \ref{ea} {\small $n+1$} times: $B^{(-\emptyset)},B^{(-1)},\ldots,B^{(-n)}$.\\
    Determine allocations: $a^{(-\emptyset)},a^{(-1)},\ldots,a^{(-n)},$ where $a^{(-i)}\in \argmax_{a \in {\mathcal F}}\widehat{V}(a|B^{(-i)}).$\\
    Pick $a^{pvm} \in \{a^{(-\emptyset)},a^{(-1)},\ldots,a^{(-n)}\}$ with maximal $\widehat{V}$.\\
    Charge each bidder $i$ according to:
    \begin{align}\label{pvmpayment}
    p_i^{pvm}:=\sum_{j\neq i}\hat{v}_j\left(a^{(-i)}_j\right)-\sum_{j\neq i}\hat{v}_j\left(a^{pvm}_j\right).
    \end{align}
    \caption{\textsc{PVM} {\small(Brero et al. 2018)}}
    \label{pvm}
\end{algorithm}

\section{Deep Neural Network-powered ICA}
In this section, we present the high level design of our DNN-powered ICA and discuss its advantages compared to the SVR-based design by \cited{brero2018combinatorial}.

Observe that the choice of the ML algorithm $\mathcal{A}$ affects Algorithm 1 
in two  ways: first, in the estimation step (Line 4), $\mathcal{A}$ determines how well we can predict bidders' valuations; second, in the optimization step (Line 5), it determines the complexity of the ML-based WDP. 
Thus, our situation is different from standard supervised learning because of the added optimization step. In particular, when choosing $\mathcal{A,}$ we must also ensure  that we obtain a practically feasible ML-based WDP.
Given that we have to solve the optimization step hundreds of times throughout an auction, in practice, we must impose a time limit on this step. In our experiments (Section \ref{ExperimentalDesign}), we follow \cited{brero2018combinatorial} and impose a 1 hour time limit on this step.

To make the optimization step feasible, \cited{brero2018combinatorial} used SVRs with quadratic kernels, for which the ML-based WDP is a quadratic integer program (QIP) and still practically solvable within a 1 hour time limit  for most settings. However, note that a quadratic kernel, while more expressive than a linear kernel, can still at most model two-way interactions between the items. To this end, \citeauthor{brero2017probably} \shortcite{brero2017probably,brero2019workingpaper} also evaluated more expressive kernels (gaussian and exponential). Even though these kernels had good prediction performance, the corresponding ML-based WDPs were too complex such that they always timed out and had a large optimization gap, thus leading to worse economic efficiency than the quadratic kernel. However, using SVRs with quadratic kernels leaves room for improvement, since bidders' valuations can be more complex than can be captured by quadratic kernels. 

In this work, we show how these shortcomings can be addressed by using DNNs instead of SVRs in the estimation and optimization steps of Algorithm 1.
DNNs are a concatenation of affine and non-linear mappings (see Figure \ref{fig:nn_tik}). They consist of several layers, which are themselves composed of multiple nodes. Between each of the layers an affine transformation is applied, which is followed by a nonlinear mapping called the \textit{activation function}.

One advantage of DNNs compared to (nonlinear) kernelized SVRs is that, for any number of layers and nodes, we always obtain a (linear) MIP for the DNN-based WDP whose size only grows linearly in the number of bidders and items (as we will show in Section \ref{MIPFormulation}). The key insight for this is to use \textit{rectified linear units (ReLUs)} as activation functions. Furthermore, in contrast to SVRs, DNNs do not use predefined feature transformations. While with SVRs, the choice of a good kernel usually relies on prior domain knowledge, DNNs automatically learn features in the process of training.

Following \cited{brero2018combinatorial}, we   decompose the estimated social welfare function in Line 4 of Algorithm \ref{ea}, $\tilde{V}^t=\mathcal A (B^{t-1})$, as follows:
\begin{align}
\tilde{V}^t=\sum_{i \in N}\tilde{v}_i^t,
\end{align}
where $\tilde{v}_i^t$ is an estimate of bidder $i$'s true value function $v_i$ and is trained on the data set $B_i^{t-1}$, i.e., the values queried up to round $t-1$. In this work, for every $i \in N$, we model $\tilde{v}^t_i$ using a \textit{fully connected feed-forward DNN} ${\mathcal{N}_i:\{0,1\}^m\to \mathbb{R}_+}$. 
Consequently, the estimated social welfare function $\tilde{V}^t$ is given as a sum of DNNs, i.e.,
\begin{align}
\tilde{V}^t:=~\sum_{i \in N}\mathcal{N}_i.
\end{align}
Note that each bidder's value function is modeled as a distinct DNN (with different architectures and parameters) because bidders' preferences are usually highly idiosyncratic.\footnote{A second reason for this construction is that this prevents bidders from influencing each others' ML-models. Please see \cite{brero2019workingpaper} for a detailed incentive analysis of PVM.}

\section{MIP Formulation of the DNN-based Winner Determination Problem}\label{MIPFormulation}
We now present the parameterization of each DNN and show how to reformulate the DNN-based WDP into a MIP. Thus, we focus on the optimization step (Line 5) of Algorithm \ref{ea}.

\subsection{Setting up the DNN-based WDP}
Figure \ref{fig:nn_tik} shows a schematic representation of a DNN $\mathcal{N}_i$. We now define the parameters of the DNNs. To simplify the exposition, we consider a fixed iteration step $t$ and no longer highlight the dependency of all variables on $t$.

Each DNN $\mathcal{N}_i$ consists of $K_i-1$ \textit{hidden layers} for $K_i \in \mathbb{N}$, with the $k$\textsuperscript{th} hidden layer containing $d^i_k$ hidden nodes, where $k\in \{1,\ldots,K_i-1\}$. As $\mathcal{N}_i$ maps bundles $x\in \{0,1\}^m$ to values, the dimension of the \textit{input layer} $d^i_0$ is equal to the number of items $m$ (i.e., $d^i_0:=m$) and the dimension of the \textit{output layer} is equal to 1 (i.e.,  $d^i_{K_i}:=1$). Hence, in total, a single DNN consists of $K_i+1$ layers. Furthermore, let $\varphi:\mathbb{R}\to \mathbb{R}_+$, $\varphi(s):=\max(0,s)$ denote the ReLU activation function.\footnote{$\varphi$ acts on vectors componentwise, i.e., for $s \in \mathbb{R}^k$ let ${\varphi(s):=(\varphi(s_1),\ldots,\varphi(s_k))}$.} The affine mappings between the $k$\textsuperscript{th} and the $k$+$1$\textsuperscript{st} layer are parameterized by a matrix $W^{i,k} \in \mathbb{R}^{d^i_{k+1}\times d^i_{k}}$ and a bias $b^{i,k} \in \mathbb{R}^{d^i_{k+1}}, \, k \in \{0,\ldots,K_i-1\}$.

To estimate the parameters $W^{i,k}$ and $b^{i,k}$ from data (i.e., from the bundle-value pairs $B_i$) we use the \textit{ADAM} algorithm, which is a popular gradient-based optimization technique \cite{kingma2014adam}.\footnote{For fitting the DNNs in all of our experiments we use \textsc{Python 3.5.3}, \textsc{Keras 2.2.4} and \textsc{Tensorflow 1.13.1}.} This is done in the \textit{estimation step} in Line 4 of Algorithm \ref{ea}. Thus,  after the estimation step, $W^{i,k}$ and $b^{i,k}$ are constants. 
In summary, given estimated parameters $W^i:=\{W^{i,k}\}_{0\le k \le K_i-1}$ and $b^i:=\{b^{i,k}\}_{0\le k \le K_i-1}$, each DNN ${\mathcal{N}_i(W^i,b^i):\{0,1\}^m \to \mathbb{R}_+}$ represents the following nested function:
\begin{align}\label{NN}
&\mathcal{N}_i(W^i,b^i)(x)=\\
&=\varphi\left(W^{i,K_i-1}\varphi\left(\ldots\varphi(W^{i,0}x+b^{i,0})\ldots\right)+b^{i,K_i-1}\right).\notag
\end{align}
The \textit{DNN-based WDP} in the optimization step in Line 5 of Algorithm \ref{ea} can now be formulated as follows:
\begin{align}\label{P1math}
&\max\limits_{a \in \mathcal{X}^n}\left\{\sum_{i \in N}\mathcal{N}_i\left(W^i,b^i\right)(a_i)\right\}\tag{OP1}\\
\textrm{s.t.}\quad&\sum_{i \in N} a_{ij}\le 1,\quad \forall j \in M\notag\\
&a_{ij}\in \{0,1\},\quad \forall j \in M, \forall i \in N.\notag
\end{align}
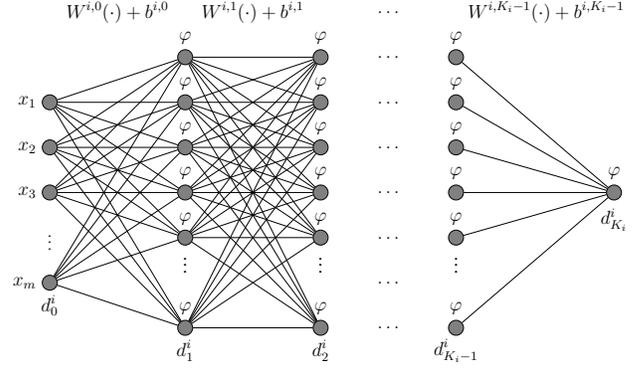
\begin{figure}[t!]
        \resizebox{\columnwidth}{!}{
                \begin{tikzpicture}
                [cnode/.style={draw=black,fill=#1,minimum width=3mm,circle}]
                \node[cnode=gray,label=90:$\mathlarger{\mathlarger{\varphi}}$] (s) at (12.5,-3) {};
                \node at (0,-4) {$\vdots$};
                \node at (3,-4.5) {$\mathlarger{\mathlarger{\vdots}}$};
                \node at (6,-4.5) {$\mathlarger{\mathlarger{\vdots}}$};
                \node at (9,-4.5) {$\mathlarger{\mathlarger{\vdots}}$};
                \node at (7.5,0) {$\mathlarger{\mathlarger{\ldots}}$};
                \node at (7.5,-1) {$\mathlarger{\mathlarger{\ldots}}$};
                \node at (7.5,-2) {$\mathlarger{\mathlarger{\ldots}}$};
                \node at (7.5,-3) {$\mathlarger{\mathlarger{\ldots}}$};
                \node at (7.5,-4) {$\mathlarger{\mathlarger{\ldots}}$};
                \node at (7.5,-5) {$\mathlarger{\mathlarger{\ldots}}$};
                \node at (7.5,-6) {$\mathlarger{\mathlarger{\ldots}}$};
                \node at (0,-5.5) {$\mathlarger{\mathlarger{d^i_0}}$};
                \node at (3,-6.5) {$\mathlarger{\mathlarger{d^i_1}}$};
                \node at (6,-6.5) {$\mathlarger{\mathlarger{d^i_2}}$};
                \node at (9,-6.5) {$\mathlarger{\mathlarger{d^i_{K_i-1}}}$};
                \node at (12.5,-3.6) {$\mathlarger{\mathlarger{d^i_{K_i}}}$};
                \node at (1.5,1) {$\mathlarger{\mathlarger{W^{i,0}(\cdot)+b^{i,0}}}$};
                \node at (4.5,1) {$\mathlarger{\mathlarger{W^{i,1}(\cdot)+b^{i,1}}}$};
                \node at (7.5,1) {$\mathlarger{\mathlarger{\ldots}}$};
                \node at (11,1) {$\mathlarger{\mathlarger{W^{i,K_i-1}(\cdot)+b^{i,K_i-1}}}$};
                \foreach \a in {1,...,4}
                {   \pgfmathparse{\a<4 ? \a : "m"}
                        \node[cnode=gray,label=180:$\mathlarger{\mathlarger{x_{\pgfmathresult}}}$]
                        (a-\a) at (0,{-\a-div(\a,4)}) {};
                }
                \foreach \a in {1,...,6}
                {       \pgfmathparse{\a<6 ? \a : "d^i_1"}
                        \node[cnode=gray,label=90:$\mathlarger{\mathlarger{\varphi}}$] (p-\a) at (3,{-(\a-1)-div(\a,6)}) {};
                }
                \foreach \a in {1,...,6}
                {       \pgfmathparse{\a<6 ? \a : "d^i_2"}
                        \node[cnode=gray,label=90:$\mathlarger{\mathlarger{\varphi}}$] (s-\a) at (6,{-(\a-1)-div(\a,6)}) {};
                }
                \foreach \a in {1,...,6}
                {       \pgfmathparse{\a<6 ? \a : "d^i_{K_i-1}"}
                        \node[cnode=gray,label=90:$\mathlarger{\mathlarger{\varphi}}$] (f-\a) at (9,{-(\a-1)-div(\a,6)}) {};
                        \draw (f-\a) -- (s);
                }
                \foreach \a in {1,...,4}
                {   \foreach \y in {1,...,6}
                        {   \draw (a-\a) -- (p-\y);
                        }
                }
                \foreach \a in {1,...,6}
                {   \foreach \y in {1,...,6}
                        {       \draw (p-\a) -- (s-\y);
                        }
                }
                
                \end{tikzpicture}
        }
        \caption{Schematic representation of a DNN $\mathcal{N}_i.$}
        \label{fig:nn_tik}
\end{figure}

\subsection{The MIP Formulation}
In its general form, \eqref{P1math} is a nonlinear, non-convex optimization problem and there do not exist practically feasible algorithms that are guaranteed to find a globally optimal solution. Therefore, we now reformulate \eqref{P1math} into a MIP.

Consider bidder $i\in N$. For every layer $k\in \{1,\ldots,K_i\}$ let $o^{i,k} \in \mathbb{R}_+^{d^i_{k}}$ denote the output of the $k$\textsuperscript{th} layer, which can be recursively calculated as
\begin{align}\label{layer_outpout}
o^{i,k}&=\varphi(W^{i,k-1}o^{i,k-1}+b^{i,k-1})=\notag\\
&=\max(0,W^{i,k-1}o^{i,k-1}+b^{i,k-1}).
\end{align}

For $k\in \{1,\ldots,K_i\}$, we introduce $d^i_k$ binary decision variables that determine which node in the corresponding layer is active, represented as a vector $y^{i,k}\in \{0,1\}^{d^i_k}$. We also introduce $2d^i_k$ continuous variables, represented as vectors $z^{i,k},\,s^{i,k}\in \mathbb{R}^{d^i_k}$. Each $z^{i,k}$ corresponds to the positive components of the output value $o^{i,k}$ of each layer and each $s^{i,k}$ is used as a slack variable representing the absolute value of the negative components of $o^{i,k}$. 

In our final MIP formulation, we will make use of \textit{``big-M''} constraints. For our theoretical results to hold, we need to make the following standard assumption.
\begin{assumption}{\small\textsc{(Big-M constraint)}}\label{bigM}
For all $i\in N$ and $k \in \{1,...,K_i\}$ there exists a large enough constant ${L\in \mathbb{R}_+}$, such that $\left|(W^{i,k-1}o^{i,k-1}+b^{i,k-1})_j\right|\le L$ for $1\le j\le d^i_k$.
\end{assumption}

In the following lemma, we show how to recursively encode a layer of $\mathcal{N}_i$ given the output value of the previous layer as multiple linear constraints.\footnote{In the following, all constraints containing vectors are defined componentwise.} 
\begin{lemma}\label{lemma1}
         Let $W^{i,k-1}o^{i,k-1}+b^{i,k-1} \in \mathbb{R}^{d^i_k}$ be fixed for a $k \in \{1,\ldots,K_i\}$. Furthermore, let $z^{i,k}, s^{i,k}\in\mathbb{R}^{d^i_{k}}$, ${y^{i,k}\in \{0,1\}^{d^i_{k}}}$. Consider the following linear constraints:
        \begin{align}
        &z^{i,k}-s^{i,k}=W^{i,k-1}o^{i,k-1}+b^{i,k-1}\label{imp00}\\
        &0\le z^{i,k}\le y^{i,k}\cdot L \label{imp11}\\
        &0\le s^{i,k} \le (1-y^{i,k}) \cdot L.\label{imp22}
        \end{align}
        The polytope defined by \eqref{imp00}-\eqref{imp22} is not empty and every element $\left(z^{i,k},s^{i,k},y^{i,k}\right)$ of this polytope satisfies  $z^{i,k}=o^{i,k}$. 
\end{lemma}
\begin{proof}
        For notational convenience, let $c:=W^{i,k-1}o^{i,k-1}+b^{i,k-1}$. Non-emptiness follows since $|c_j|\le L,\, 1\le j\le d^i_k$, by Assumption \ref{bigM}. From Constraints \eqref{imp11} and \eqref{imp22} it follows, for $1\le j \le d^i_k$, that if $z^{i,k}_j>0$ then $s^{i,k}_j=0$, and if  $s^{i,k}_j>0$ then $z^{i,k}_j=0$. We now have to distinguish the following three cases for each component of $c$:
        \begin{align*}
        &c_j<0 \implies y^{i,k}_j=0,\,s^{i,k}_j=-c_j, z^{i,k}_j=0=\varphi(c_j)\\
        &c_j>0 \implies y^{i,k}_j=1,\, s^{i,k}_j=0, z^{i,k}_j=c_j=\varphi(c_j)\\
        &c_j=0 \implies \left(z^{i,k}_j,s^{i,k}_j,y^{i,k}_j\right)\in \{(0,0,0),(0,0,1)\}
        \end{align*}
        Combining all cases yields that $z^{i,k}=\varphi(c_j)=o^{i,k}$.
\end{proof}

Given Lemma \ref{lemma1}, we can now reformulate the DNN-based WDP as a MIP. For this, we let $W^i,\, b^i$ denote the estimated parameters corresponding to $\mathcal{N}_i(W^i,b^i)$. Furthermore, let $a\in \mathcal{X}^n, y^{i,k}\in \{0,1\}^{d^i_k}$ and $z^{i,k},\,s^{i,k}\in \mathbb{R}^{d^i_k},$ for ${1\le k \le K_i}$ and $L$ be a constant satisfying Assumption \ref{bigM}.

\begin{align}\label{MIP}
&\hspace{1cm}\max\limits_{a\in \mathcal{X}^n, z^{i,k},s^{i,k},y^{i,k}}\left\{\sum_{i \in N}z^{i,K_i}\right\}\tag{OP2}\\
&\textrm{s.t.}\notag\\
&\begin{rcases}
&z^{i,0}=a_i\notag\\
&z^{i,k}-s^{i,k}=W^{i,k-1}z^{i,k-1}+b^{i,k-1}\notag\\
&0\le z^{i,k}\le y^{i,k}\cdot L \notag\\
&0\le s^{i,k} \le (1-y^{i,k}) \cdot L \notag\\
&y^{i,k}\in \{0,1\}^{d^i_k}
&\hspace{-0.38cm}\end{rcases}\hspace{-0.55cm}\begin{split}&\hspace{-0.1cm}\forall i\in N\notag\\&\hspace{-0.1cm} \forall k\in\{1,\ldots,K_i\}\end{split}\\
&\hspace{0.1cm}a_{ij}\in \{0,1\},\qquad \forall j\in M,\,\forall i \in N\notag\\
&\hspace{0.08cm}\sum_{i \in N} a_{ij}\le 1,\qquad \forall j\in M\notag
\end{align}

We are now ready to state our main theorem.

\begin{theorem}{\textsc{(MIP Formulation)}}\label{MIPTHM}
The DNN-based WDP as defined in \eqref{P1math} is equivalent to the MIP defined in \eqref{MIP}.\!     
\end{theorem}
\begin{proof}
Consider \eqref{P1math}. For each bidder $i \in N$, we first set $z^{i,0}$ equal to the input bundle $a_i$. Then we proceed by using Lemma \ref{lemma1} for $k=1$, i.e., we reformulate the output of the $1$\textsuperscript{st} layer as the linear constraints \eqref{imp00}, \eqref{imp11} and \eqref{imp22}. We iterate this procedure until we have reformulated the final layer, i.e, $k=K_i$. Doing so for each bidder $i\in N$ and adding the feasibility constraints yields \eqref{MIP}.
\end{proof}
Using the MIP formulation \eqref{MIP}, we can solve the DNN-based WDP using standard optimization packages like \textsc{Cplex}.\footnote{For solving MIPs of the form \eqref{MIP} in our experiments we use \textsc{Cplex 12.8.0.0 } with the python library \textsc{DOcplex 2.4.61}.}

We now provide a simple worked example for how to reformulate \eqref{P1math} into \eqref{MIP} which illustrates Theorem 1.
\begin{example}
Consider a setting with one bidder ($n=1$), $m$ items ($d^1_0=m$) and one hidden layer ($K_1=2$). Given $W^{1,0}\in\mathbb{R}^{d^1_1\times m},\, W^{1,1} \in \mathbb{R}^{1\times d^1_1}$ ,$b^{1,0}\in\mathbb{R}^{d^1_1},$ and $b^{1,1} \in \mathbb{R}$, \eqref{P1math} can be written as
    \begin{align*}
    \max\limits_{a_1\in \mathcal{X}^1}&\left\{\max\left(0,W^{1,1}\max\left(0,W^{1,0}a_1+b^{1,0}\right)+b^{1,1}\right)\right\}\\
    \textrm{s.t.}\quad& a_{1j}\in \{0,1\},\quad \forall j \in M,
    \end{align*}
         where the constraint $\sum_{i \in N}a_{ij}\le1,\, \forall j \in M$ is redundant in this case, since we only consider one bidder. First, we replace the inner maximum using $y^{1,1}\in \{0,1\}^{d^1_1}, z^{1,1},s^{1,1}\in \mathbb{R}^{d^1_1}$ and arrive at the equivalent optimization problem
    \begin{align*}
    &\max\limits_{\substack{a_1\in \mathcal{X}^1\\z^{1,1},\, s^{1,1},\, y^{1,1}}}\left\{\max\left(0,W^{1,1}z^{1,1}+b^{1,1}\right)\right\}\\
    \textrm{s.t.}\quad&z^{1,1}-s^{1,1}=W^{1,0}a_1+b^{1,0}\\
    &0\le z^{1,1}\le y^{1,1}\cdot L\\
    &0\le s^{1,1} \le (1-y^{1,1}) \cdot L\\
    &a_{1j}\in \{0,1\},\quad \forall j \in M.
    \end{align*}
    Applying Lemma 1 again by using $y^{1,2}\in \{0,1\}$ and $z^{1,2},s^{1,2}\in \mathbb{R}$ yields the final MIP formulation
    \begin{align*}
    &\max\limits_{\substack{a_1\in \mathcal{X}^1,\, z^{1,k}\\y^{1,k},\, s^{1,k},\, k\,\in\,\{1,2\}}}\left\{z^{1,2}\right\}\\
    \textrm{s.t.}\quad&z^{1,1}-s^{1,1}=W^{1,0}a_1+b^{1,0}\\
    &0\le z^{1,1}\le y^{1,1}\cdot L\\
    &0\le s^{1,1} \le (1-y^{1,1}) \cdot L\\
    &z^{1,2}-s^{1,2}=W^{1,1}z^{1,1}+b^{1,1}\\
    &0\le z^{1,2}\le y^{1,2}\cdot L\\
    &0\le s^{1,2} \le (1-y^{1,2}) \cdot L\\
    &a_{1j}\in \{0,1\},\quad \forall j \in M.
    \end{align*}
\end{example}
\begin{remark}\label{RemarkMIP}
The number of decision variables in the MIP defined in \eqref{MIP} is given by {\small
$$\underbrace{\sum_{i \in N}}_{\text{\#bidders}}\left(\underbrace{m}_{\text{\#items}}+ \underbrace{3}_{y^{i,k},\, s^{i,k},  z^{i,k}}\cdot\left(\sum_{k=1}^{\overbrace{K_i-1}^{\text{\#hidden layers}}}\underbrace{d^i_k}_{\text{\#nodes per layer}}+\overbrace{1}^{\text{output}}\right)\right)$$}
\end{remark}

\section{Experimental Evaluation}\label{ExperimentalDesign}
In this section, we evaluate the performance of our deep
learning-powered ICA and compare it against the SVR-based
approach using quadratic kernels by \cited{brero2018combinatorial}. We release our code under an open-source license at: https://github.com/marketdesignresearch/DL-ICA.

\subsection{Experiment setup}
Spectrum auctions are one of the most prominent applications of CAs, which is why we choose them for our experiments. Specifically, we use the spectrum auction test suite (SATS) version 0.6.4 \cite{weiss2017sats}.\footnote{Experiments were conducted on machines with Intel Xeon E5-2650 v4 2.20GHz processors with 20 cores.} SATS enables us to generate 1000s of CA instances in different domains. Furthermore, we have access to each bidder's true value  $v_i(x)$ for all $2^m$ possible bundles $x\in \mathcal{X}$ as well as the efficient allocation $a^*\in \mathcal{F}$, which we can use to measure the efficiency of any other allocation $a$ by $V(a)/V(a^*)$. We evaluate our approach in the following three domains:

\textit{The Global Synergy Value Model} (GSVM) \cite{goeree2010hierarchical} consists of $6$ \textit{regional bidders}, $1$ \textit{national bidder}, and $18$ items. In GSVM the value of a package increases by a certain percentage with \textit{every} additional item of interest. Thus, the value of a bundle only depends on the total number of items contained in a bundle which makes it one of the simplest models in SATS. In fact, bidders' valuations can be exactly learned by SVRs with quadratic kernels (\cited{brero2019workingpaper}) which implies that the valuations exhibit at most two-way interactions between items. 

\textit{The Local Synergy Value Model} (LSVM) \cite{scheffel2012impact} consists of $5$ \textit{regional bidders}, $1$ \textit{national bidder} and $18$ items. The items are arranged on a rectangle of size $3\times6$. The national bidder is interested in all items, while the regional bidders are only interested in certain subsets of items. Complementarities  arise from spatial proximity of items and are modeled via a logistic function, which makes it more complex than GSVM.

The \textit{Multi-Region Value Model} (MRVM) \cite{weiss2017sats} consists of $98$ items and $10$ bidders. It models large US and Canadian spectrum auctions and captures both geography (different regions) as well as frequency dimensions (different bands). A bidder
is categorized as \textit{national}, \textit{regional} or \textit{local}, depending on the magnitude of the synergies between different regions.

In Sections \ref{predperformance} and \ref{effresults}, we first evaluate our approach in detail using the two medium-sized domains GSVM and LSVM. Then, in Section \ref{scalingtolargerdomains}, we use MRVM to evaluate how well our approach scales to very large domains.

\subsection{Prediction Performance}\label{predperformance}
We first compare the \textit{prediction performance} of DNNs to SVRs. Using SATS, we generate a data set of bundle-value pairs $\{(x^{(k)},v_i(x^{(k)}))\}$ for all bidders $i \in N$.
For each auction domain, we draw $100$ auction instances uniformly at random. For each such instance, we sample, for each bidder type, a training set $T$ of equal size and a disjoint test set $V$ consisting of all remaining bundles, i.e., $|V|:=2^{|M|}-|T|$. For each bidder type, we train the ML algorithm on $T$ and test it on $V$. We report the mean absolute error (MAE) for both bidder types averaged over the $100$ instances.\footnote{For training the DNNs, we use the MAE as the loss function.}

We denote by $[d_1,d_2,d_3]$ a 3-hidden-layer DNN with $d_1$, $d_2$ and $d_3$ hidden nodes, respectively. For both, SVRs and DNNs, we performed a hyperparameter optimization for each bidder type. For the DNNs, we optimized the architecture\footnote{We considered the following architectures: for the national bidders: [10], [100], [10,10], $\ldots$, [100,100,100,100], and for the regional bidders: [32], [100], [32,32], $\ldots$, [100,100,100,100].}, the L2-regularization parameter for the affine mappings, the dropout rate per layer, and the learning rate of ADAM. For SVRs with a quadratic kernel $k(x,y):=x^Ty+\gamma(x^Ty)^{2}$, we optimized $\gamma$ (i.e., the influence of the quadratic term),  the regularization parameter $C$, and the loss function parameter $\epsilon$. In what follows, we present the winner models resulting from this hyperparameter optimization.

\begin{table}[t!]
        \centering
        \resizebox{1\columnwidth}{!}{
                \begin{tabular}{cccccc}
                        \toprule
                        && Bidder & DNN & & \\
                        ML Algorithm  & $|T|$ & Type & Architecture & MAE\textsubscript{train} & MAE\textsubscript{test}\\
                        \midrule
                        & \multirow{2}{*}{50}  &  National & [100] &  1.99&  \textbf{5.25} {\footnotesize(0.11)}\\
                        &  &   Regional & [100] & 2.06 &  \textbf{6.20} {\footnotesize(0.19)} \\
                        \cmidrule{2-6}
                        \multirow{2}{*}{DNNs} & \multirow{2}{*}{100} &  National & [100] & 2.10& \textbf{3.66} {\footnotesize(0.07)}\\
                        & &  Regional & [100] &  2.71&  \textbf{4.64} {\footnotesize(0.13)}\\
                        \cmidrule{2-6}
                        & \multirow{2}{*}{200} &  National & [100]& 1.61 &  \textbf{2.22} {\footnotesize(0.05)}\\
                        &  &  Regional & [100] &2.11&  \textbf{2.89} {\footnotesize(0.09)} \\                                                                       
                        \midrule
                        \midrule
                        & & & Kernel & & \\
                        \midrule
                        & \multirow{2}{*}{50}&  National & quadratic & 0.03&  \textbf{4.38} {\footnotesize(0.11)}\\
                        &  &  Regional & quadratic  & 0.03 &  \textbf{4.98} {\footnotesize(0.20)}\\
                        \cmidrule{2-6}
                        \multirow{2}{*}{SVRs} &\multirow{2}{*}{100}&  National &quadratic  & 0.03 &  \textbf{1.71} {\footnotesize(0.04)}\\
                        & &  Regional & quadratic & 0.03&  \textbf{2.07} {\footnotesize(0.07)}\\
                        \cmidrule{2-6}
                        &\multirow{2}{*}{200}&  National & quadratic& 0.03&  \textbf{0.12} {\footnotesize(0.00)}\\
                        &  &  Regional & quadratic  & 0.03&  \textbf{0.13} {\footnotesize(0.00)}\\
                        \bottomrule
                \end{tabular}   
        }
        \caption{Prediction performance in GSVM. All results are averaged over
100 auction instances. For MAE\textsubscript{test}, standard errors are shown in parentheses.}
\label{ESTGSVM}
\end{table}

In Table \ref{ESTGSVM}, we present  prediction performance results for the GSVM domain.  Consider the last column of the table, which shows the MAE on the test set. We observe the very good prediction performance of the SVRs and in particular that the test error converges to $0$ when increasing $|T|$. This is due to the fact that in GSVM, bidders' value functions can be perfectly captured by quadratic kernels. In this sense, GSVM represents a ``worst case'' auction domain w.r.t. our comparison of DNNs against quadratically-kernelized SVRs. Looking at the performance of the DNNs, we observe that the test error also decreases with $|T|$, but, not surprisingly, is always larger than for the SVRs with quadratic kernels. Furthermore, we observe that the optimal architectures are always a $1$-hidden layer network.

\begin{table}[htbp]
\centering
\resizebox{1\columnwidth}{!}{
        \begin{tabular}{cccccc}
                \toprule
                &  & Bidder & DNN &   & \\
                ML Algorithm  & $|T|$ & Type & Architecture & MAE\textsubscript{train} & MAE\textsubscript{test}\\
                \midrule
                & \multirow{2}{*}{50}  &  National  & [10] & 24.68&  \textbf{29.90} {\footnotesize(0.23)}\\
                &  &   Regional & [100] & 4.22 &  \textbf{16.58} {\footnotesize(0.39)} \\
                \cmidrule{2-6}
                \multirow{2}{*}{DNNs} & \multirow{2}{*}{100} &  National & [10, 10, 10] & 9.01 & \textbf{25.62} {\footnotesize(0.36)}\\
                &        &  Regional &  [100] & 5.01&  \textbf{13.74} {\footnotesize(0.26)}\\
                \cmidrule{2-6}
                & \multirow{2}{*}{200} &  National & [10, 10] & 10.52&  \textbf{20.58} {\footnotesize(0.21)}\\
                &  &  Regional &  [100, 100, 100] & 3.64&  \textbf{11.27} {\footnotesize(0.23)} \\                                                                       
                \midrule
                \midrule
                & & & Kernel & & \\
                \midrule
                &  \multirow{2}{*}{50} &  National & quadratic  & 18.51&  \textbf{32.61} {\footnotesize(0.59)}\\
                &                                          &  Regional & quadratic & 3.11 &  \textbf{15.30} {\footnotesize(0.34)}\\
                \cmidrule{2-6}
                \multirow{2}{*}{SVRs} &\multirow{2}{*}{100}&  National & quadratic&  20.03 &  \textbf{27.86} {\footnotesize(0.28)}\\
                &  &  Regional & quadratic & 3.21 &  \textbf{14.21} {\footnotesize(0.28)}\\
                \cmidrule{2-6}
                &\multirow{2}{*}{200}&  National & quadratic& 20.03 &  \textbf{25.44} {\footnotesize(0.16)}\\
                &  &  Regional & quadratic&  8.23 &  \textbf{12.67} {\footnotesize(0.26)}\\
                \bottomrule
\end{tabular}   
}
\caption{Prediction performance in LSVM. All results are averaged over
100 auction instances. For MAE\textsubscript{test}, standard errors are shown in parentheses.}
\label{ESTLSVM}
\end{table}

In Table \ref{ESTLSVM}, we present the results for the more complex LSVM domain. We observe that for $|T|=50$, the DNNs and SVRs with quadratic kernels have similar test error. But for $|T|=100$ and $|T|=200$, the DNNs significantly outperform the SVRs with quadratic kernels. Specifically, DNNs better capture the national bidder in LSVM, which is important, since this bidder is interested in all items and usually gets a large portion of the items in the final allocation, which matters a lot for efficiency. In contrast to GSVM, we observe that for $|T|\ge100$, multi-hidden-layer networks were found to be best. This suggests that DNNs may indeed be advantageous for capturing more complex preference structures.\footnote{We also evaluated the prediction performance of other kernels (linear, gaussian and exponential) and observed that DNNs were as good or better for almost all combinations of bidder types and $|T|$.} 

\subsection{Efficiency Results}\label{effresults}
Finally, we compare the economic efficiency of our DNN-powered ICA against the SVR-powered ICA. When conducting the efficiency experiments, we follow \cited{brero2018combinatorial} and assume that bidders answer all value queries truthfully (i.e., $\hat{v}_i=v_i$). Furthermore, we also use their experiment setup and define a cap $c_e$ on the total number of value queries in Algorithm \ref{ea} and set $c_e:=50$. The initial set of reported bundle-value pairs $B^0_i$ per bidder $i$ is drawn uniformly at random. We denote the number of initial reports by $c_0:=|B^0_i|, \, \forall i \in N$, resulting in a maximum of $c_0 + n\cdot(c_e-c_0)$ queries per bidder.

\subsubsection{GSVM.}
In Table \ref{EFFGSVMTRAIN}, we first present the results from comparing nine different network architectures on a training set of $100$ GSVM instances. As we can see, the winning model is among the largest multi-layer networks we tested. It is noteworthy that the one-hidden-layer network (R:[100]$|$N:[100]), which performed best in terms of prediction performance, did not perform best in terms of efficiency.  

In Table \ref{EFFGSVMTEST}, we compare the performance of the winner model from Table \ref{EFFGSVMTRAIN} (see Appendix~A for configuration details) against the SVR-based approach on a test set of $100$ GSVM instances. Even though GSVM can be perfectly captured by quadratic kernels, our DNN-based approach achieves a similar result w.r.t. efficiency, where the difference in means is not statistically significant ($p=0.337$).
\begin{table}[t!]
        \resizebox{\columnwidth}{!}{
                \setlength\tabcolsep{2.5pt}
                \begin{tabular}{lcccc}
                        \toprule
                        \textbf{DNN Architectures}\footnotemark &  $\boldsymbol{c_0}$ & \textbf{Efficiency \%}& \textbf{Revenue \%}\footnotemark\\
                        \midrule
                        R:$[16,16]\,\boldsymbol{|}\,$N:$[10,10]$  & 40  & 98.53\% & 69.26\% \\
                        R:$[16,16]\,\boldsymbol{|}\,$N:$[10,10]$  & 30  & 98.41\% & 68.29\% \\
                        R:$[16,16]\,\boldsymbol{|}\,$N:$[10,10,10]$  & 40  & 98.51\% & 68.91\% \\
                        R:$[16,16]\,\boldsymbol{|}\,$N:$[10,10,10]$  & 30 & 98.32\% & 68.59\% \\
                        R:$[32,32]\,\boldsymbol{|}\,$N:$[10,10]$  & 40  & 98.75\% & 71.14\% \\
                        \textbf{R:[32,32] $\boldsymbol{|}$ N:[10,10]}  & \textbf{30}  & \textbf{98.94\%} & \textbf{68.47\%} \\ 
                        R:$[32,32]\,\boldsymbol{|}\,$N:$[10,10,10]$  & 40  & 98.69\% & 71.63\% \\
                        R:$[32,32]\,\boldsymbol{|}\,$N:$[10,10,10]$  & 30  & 98.92\% & 68.88\% \\
                        R:$[100]\,\boldsymbol{|}\,$N:$[100]$  & 30  & 98.27\% & 66.73\% \\
                        \bottomrule
                \end{tabular}
        }
        \caption{Efficiency results for $9$ configurations of a DNN-powered ICA on a training set of $100$ GSVM auction instances. The selected winner model is marked in bold. All results are averaged over the $100$ auction instances.}
        \label{EFFGSVMTRAIN}
\end{table}
\footnotetext[11]{We denote by R and N the architectures used for the regional- and national bidders, respectively. }
\footnotetext[12]{Revenue is calculated as $(\sum_{i \in N}p_i^{pvm})/V(a^*).$}
\begin{table}[t!]
        \resizebox{\columnwidth}{!}{
                \setlength\tabcolsep{2.5pt}
                \begin{tabular}{cccccc}
                        \toprule
                         \textbf{Auction} &  & \textbf{Max} & \%& \%& \textbf{t-test on} \\
                        \textbf{Mechanism} & \textbf{\#Queries} & \textbf{\#Queries} & \textbf{Efficiency} &\textbf{Revenue}  &\textbf{Efficiency\footnotemark} \\
                        \midrule
                        VCG & $2^{18}$ & $2^{18}$ & 100.00 {\footnotesize(0.00)} & 80.4 & -\\
                        \midrule
                         SVR-ICA & 41.9 & 42.8 & \textbf{98.85} {\footnotesize(0.13)} & 77.80& \multirow{2}{*}{0.337}\\
                        \cmidrule{1-5}
                         DNN-ICA & 53 & 78 & \textbf{98.63} {\footnotesize(0.18)} & 67.81 &\\
                        \bottomrule
                \end{tabular}
        }
        \caption{A comparison of the DNN-powered ICA against the SVR-powered ICA and VCG (as reported in \protect\cited{brero2018combinatorial}) on a test set of $100$ GSVM instances. All results are averaged over the 100 instances. For efficiency, standard errors are shown in parentheses.}
        \label{EFFGSVMTEST}
\end{table}
\footnotetext[13]{We performed a two-sided unpaired \textit{Welch Two Sample t-test}  with $\mathcal{H}_0:\mu_1=\mu_2$ against $\mathcal{H}_A:\mu_1\neq\mu_2$. We thank \cited{brero2018combinatorial} for providing us with the detailed results of their experiments to enable all t-tests reported in this paper.}

\subsubsection{LSVM.}
We now turn to the more complex LSVM domain. We first select a winner model based on a training set of $100$ LSVM instances (Table \ref{EFFLSVMTRAIN}). As in GSVM, the best model is among the largest architectures and the best model w.r.t. prediction performance does not yield the highest efficiency.

\begin{table}[t!]
        \resizebox{\columnwidth}{!}{
        \begin{tabular}{lcccc}
                \toprule
                \textbf{DNN Architectures}&  $\boldsymbol{c_0}$ & \textbf{Efficiency \%}& \textbf{Revenue \%} \\
                \midrule
                R:$[16,16]\,\boldsymbol{|}\,$N:$[10,10]$  & 40 & 97.40 & 60.51 \\
                R:$[16,16]\,\boldsymbol{|}\,$N:$[10,10]$  & 30 & 96.87 & 56.85 \\
                R:$[16,16]\,\boldsymbol{|}\,$N:$[10,10,10]$  & 40 & 97.45 & 61.15 \\
                R:$[16,16]\,\boldsymbol{|}\,$N:$[10,10,10]$  & 30 & 97.12 & 59.31 \\
                R:$[32,32]\,\boldsymbol{|}\,$N:$[10,10]$  & 40 & 97.40 & 62.01 \\
                R:$[32,32]\,\boldsymbol{|}\,$N:$[10,10]$  & 30 & 96.83 & 59.07 \\
                \textbf{R:[32,32] $\boldsymbol{|}$ N:[10,10,10]} & \textbf{40} & \textbf{97.74} & \textbf{61.95} \\
                R:$[32,32]\,\boldsymbol{|}\,$N:$[10,10,10]$  & 30 & 97.12 & 59.56 \\
                R:$[100]\,\boldsymbol{|}\,$N:$[10]$  & 40 & 96.78 & 58.71 \\
                \bottomrule
        \end{tabular}}
\caption{Efficiency results for  $9$ configurations of a DNN-powered ICA on a training set of 100 LSVM auction instances. The selected winner model is marked in bold. All results are averaged over the 100 auction instances.}
\label{EFFLSVMTRAIN}
\end{table}

\begin{table}[t!]
        \resizebox{\columnwidth}{!}{
                \setlength\tabcolsep{2.5pt}
                \begin{tabular}{ccccccc}
                \toprule
                 \textbf{Auction} &  & \textbf{Max} & \%&\% &\textbf{t-test on}\\
                \textbf{Mechanism} & \textbf{\#Queries} & \textbf{\#Queries} & \textbf{Efficiency} &\textbf{Revenue} &\textbf{Efficiency}\\
                \midrule
                VCG & $2^{18}$ & $2^{18}$ & 100.00 {\footnotesize(0.00)} & 83.4 & -\\
                \midrule
                SVR-ICA & 48.2 & 52.8 & \textbf{96.03} {\footnotesize(0.33)}& 65.60 & \multirow{2}{*}{$4\mathrm{e}{-5}$}\\
                \cmidrule{1-5}
                DNN-ICA & 65 & 77 & \textbf{97.74} {\footnotesize(0.24)} & 62.45 &\\
                \bottomrule
                \end{tabular}}
                \caption{A comparison of the DNN-powered ICA against the SVR-powered ICA and VCG (as reported in \protect\cited{brero2018combinatorial}) on a test set of $100$ LSVM auction instances. All results are averaged over the $100$ instances. For efficiency, standard errors are shown in parentheses.}
                \label{EFFLSVMTEST}
\end{table}

In Table \ref{EFFLSVMTEST}, we compare the performance of the selected winner model from Table \ref{EFFLSVMTRAIN} (see Appendix~A for configuration details) against the SVR-based approach on a test set of $100$ new auction instances. Here, we see that our DNN-powered ICA substantially outperforms the SVR-powered ICA by $1.71\%$ points, and that the difference in means is highly statistically significant ($p=4\mathrm{e}{-5}$). This demonstrates the advantage of DNNs over SVRs with quadratic kernels in complex domains like LSVM.

In Figure \ref{fig:lsvmhist}, we present a histogram of the efficiency obtained by the selected winner model on the test set.
We see that for $29$ auction instances, our approach (impressively) obtains an economic efficiency of $100\%$. However, for two instances, the efficiency is less than $90\%$. Thus, it is a promising avenue for future work to investigate these outliers to further increase the average efficiency.

\begin{figure}[t!]
        \centering
        \includegraphics[width=0.9\columnwidth]{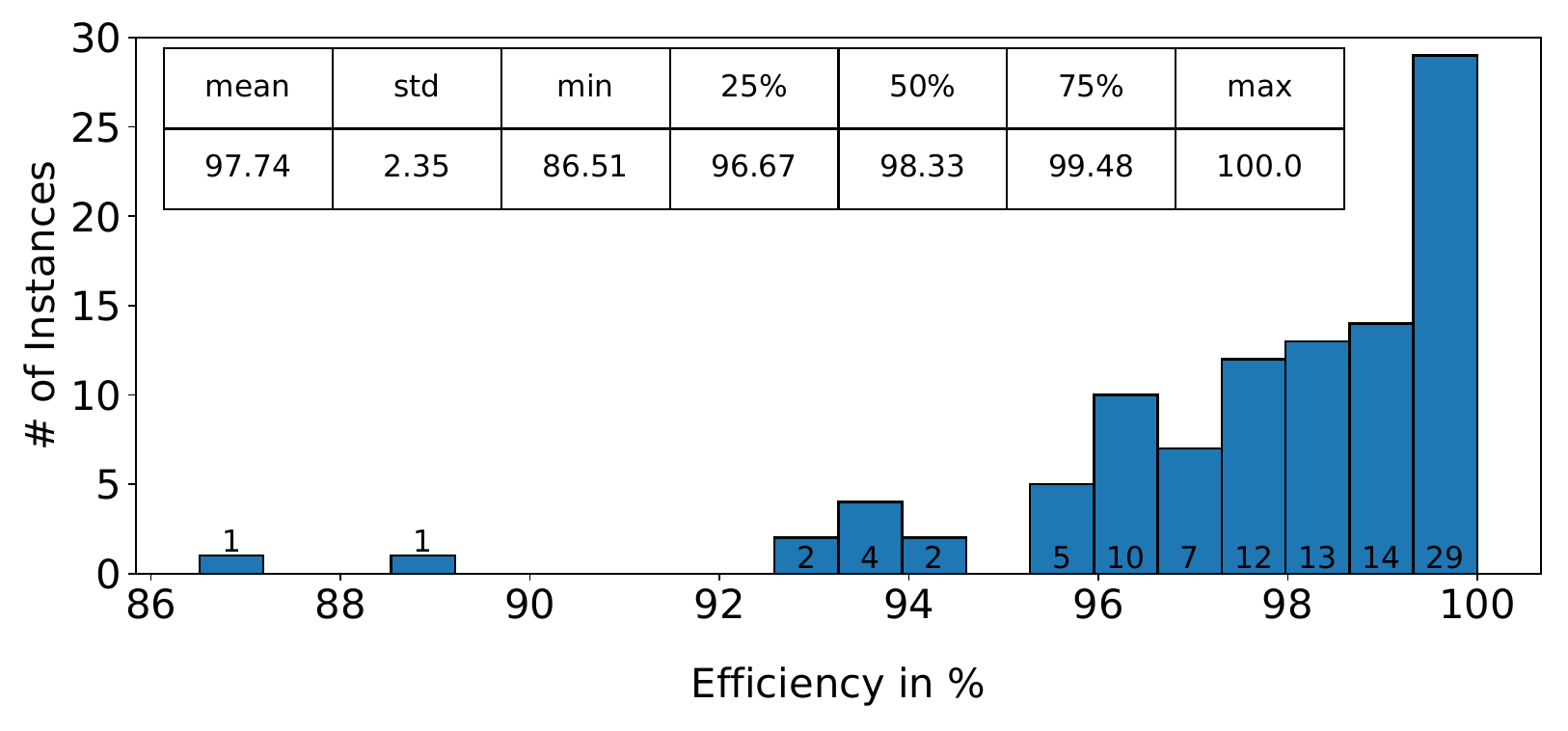}
        \caption{Histogram of efficiency results in LSVM of the selected DNN winner model on the test set.}
        \label{fig:lsvmhist}
\end{figure}

\begin{remark}
In Tables \ref{EFFGSVMTEST} and \ref{EFFLSVMTEST}, we see that our DNN-based approach achieves lower revenue than the SVR-based approach. This may be explained as follows. A bidder's payment in PVM, depends on the difference between the social welfare in the marginal and the main economy. However, PVM has one oddity: a bidder's payment may be negative. This happens more frequently with DNNs than with SVRs: consider an auction where bidder $i$ is not allocated in the main economy. Then, ideally, the allocation (and thus the welfare) should be exactly the same in the marginal economy where bidder $i$ is excluded, resulting in a zero payment. When using SVRs, this is guaranteed if the same set of bundle-value pairs is used in the main and marginal economy, because SVRs use a deterministic algorithm for estimation. In contrast, DNNs use a non-deterministic algorithm, sometimes resulting in different allocations in the main and marginal economies. However, this is more a limitation of PVM itself. In practice, one should lower bound the payments as also suggested by \cited{brero2019workingpaper}. Lower-bounding all payments by zero increases the revenue in GSVM by 7.9\% points and in LSVM by 8.3\% points.
\end{remark} 

\subsection{Scaling to Larger Domains}\label{scalingtolargerdomains}
We now present results for the MRVM domain (with 98 items and 10 bidders) to show that our DNN-powered ICA also scales well to larger domains (we present detailed runtime results in Section \ref{runtimeanalysis}).
We use the experiment setup of \cited{brero2018combinatorial} and set $c_e:=100$.

\begin{table}[b!]
        \resizebox{\columnwidth}{!}{
                \begin{tabular}{lccccc}
                        \toprule
                        \textbf{DNN Architectures} &  $\boldsymbol{c_0}$ & \textbf{Efficiency \%}& \textbf{Revenue \%} \\
                        \midrule
                        L:$[10,10]\,\boldsymbol{|}\,$R:$[32,32]\,\boldsymbol{|}\,$N:$[32,32]$  & 70 & 93.54 & 31.02\\
                        L:$[10,10]\,\boldsymbol{|}\,$R:$[32,32]\,\boldsymbol{|}\,$N:$[32,32]$  & 50 & 94.07 & 33.51\\
                        L:$[16,16]\,\boldsymbol{|}\,$R:$[16,16]\,\boldsymbol{|}\,$N:$[16,16]$  & 30 & 94.46 & 31.39\\
                        L:$[10,10]\,\boldsymbol{|}\,$R:$[32,32]\,\boldsymbol{|}\,$N:$[32,32]$  & 30 & 94.73 & 31.88\\
                        L:$[16,16]\,\boldsymbol{|}\,$R:$[16,16]\,\boldsymbol{|}\,$N:$[16,16]$  &  20 & 94.88
                         & 30.31\\
                        L:$[10,10]\,\boldsymbol{|}\,$R:$[32,32]\,\boldsymbol{|}\,$N:$[32,32]$  & 20 & 94.42
                         & 34.23\\
                        \textbf{L:[16,16]$\boldsymbol{|}$ R:[16,16] $\boldsymbol{|}$ N:[16,16]}  & \textbf{10} & \textbf{95.00}
                         & \textbf{31.97}\\
                        L:$[10,10]\,\boldsymbol{|}\,$R:$[32,32]\,\boldsymbol{|}\,$N:$[32,32]$  &10  & 94.54 & 34.78\\
                        L:$[10,10]\,\boldsymbol{|}\,$R:$[16,16,16]\,\boldsymbol{|}\,$N:$[16,16,16]$ & 10 & 94.74 & 31.12\\
                        \bottomrule
        \end{tabular}}
        \caption{Efficiency results for 9 configurations of a DNN-powered ICA on a training set of $19$ MRVM auction instances. The selected winner model is marked in bold. All results are averaged over the $19$ auction instances.}
        \label{EFFMRVMTRAIN}
\end{table}

\begin{table}[t!]
        \resizebox{\columnwidth}{!}{
                \setlength\tabcolsep{2.5pt}
                \begin{tabular}{cccccc}
                        \toprule
                        \textbf{Auction} &  & \textbf{Max} & \%&\% &\textbf{t-test on}\\
                        \textbf{Mechanism} & \textbf{\#Queries} & \textbf{\#Queries} & \textbf{Efficiency} &\textbf{Revenue} &\textbf{Efficiency}\\
                        \midrule
                        VCG & $2^{98}$ & $2^{98}$ & 100.00 {\footnotesize(0.00)} & 44.3 & -\\
                        \midrule
                        SVR-ICA & 265 & 630 & \textbf{94.58} {\footnotesize(0.14)}& 35.20 & \multirow{2}{*}{0.0268}\\
                        \cmidrule{1-5}
                        DNN-ICA & 334 & 908 & \textbf{95.04} {\footnotesize(0.14)} & 30.59 &\\
                        \bottomrule
        \end{tabular}}
        \caption{A comparison of the DNN-powered ICA against the SVR-powered ICA and VCG (as reported in \protect\cited{brero2018combinatorial}) on a test set of $50$ MRVM auction instances. All results are averaged over the $50$ instances. For efficiency, standard errors are shown in parentheses.}
        \label{EFFMRVMTEST}
\end{table}

In Table \ref{EFFMRVMTRAIN}, we present the results for different DNN architectures and different values of $c_0$, evaluated on a training set of MRVM instances. First, we observe  that the efficiency increases as we decrease $c_0$. This can be explained by the fact that a smaller $c_0$ tends to lead to more iterations of the preference elicitation algorithm, resulting in a larger number of elicited bundle-value pairs. In terms of which DNN architectures performed better or worse, no clear pattern emerged.
 
In Table \ref{EFFMRVMTEST}, we compare the performance of the  selected winner model from Table \ref{EFFMRVMTRAIN} (see Appendix~A for configuration details) against the SVR-based approach on a test set of $50$ MRVM instances. We see that our DNN-powered ICA outperforms the SVR-powered ICA by $0.46\%$ points. While this is a relatively modest increase in efficiency, a $t$-test shows that the difference in means is statistically significant $(p=0.0268)$. We also observe that our DNN-based approach (while obeying all caps) asks a larger number of queries than the SVR-based approach. It is unclear how much of the efficiency increase is driven by the DNN or by the larger number of queries. Future work should compare the two approaches by holding the total number of queries constant.

\subsection{Runtime Analysis}\label{runtimeanalysis}

In Table \ref{DNNruntime}, we  present runtime results for our DNN-powered ICA for the winner models from Tables \ref{EFFGSVMTEST}, \ref{EFFLSVMTEST} and \ref{EFFMRVMTEST}. Specifically, we show average runtime results of the MIP \eqref{MIP}, of an iteration of Algorithm \ref{ea}, and of a whole auction (PVM). We observe that the average runtime of a whole auction takes approximately $1$ hour in GSVM and LSVM and $8$ hours in the larger MRVM domain. The increase in total runtime in MRVM can be explained by the fact that we use a smaller number of initial queries ($c_0:=10$) and a larger total query cap ($c_e:=100$) compared to LSVM and GSVM. This results in a larger number of iterations of Algorithm \ref{ea}. Additionally, MRVM consists of $10$ bidders resulting in $11$ calls of Algorithm \ref{ea} in contrast to $7$ calls in LSVM and $8$ in GSVM. Even though in MRVM the average MIP runtime is smaller, the larger number of iterations and bidders lead to this increase in total runtime. Overall, these results show that our DNN-based approach is practically feasible and scales well to the larger MRVM domain. \cited{brero2018combinatorial} do not provide runtime information such that we cannot provide a runtime comparison with SVRs.\footnote{In conversations with the authors, they told us that for gaussian and exponential kernels, their MIPs always timed out (1h cap) in GSVM, LSVM and MRVM. The average MIP runtime for the quadratic kernel was a few seconds in GSVM and LSVM. In MRVM, the quadratic kernel also regularly timed out resulting in an average MIP runtime of 30 min and of 36 h for a whole auction.}

In Figure \ref{fig:miplsvm}, we present additional MIP runtime results for selected DNN architectures in LSVM (results in GSVM and MRVM are qualitatively similar). We observe two effects: First, increasing the number of nodes per layer slightly increases the average runtime. Second, adding an additional layer (for the national bidder) significantly increases the average runtime. Not surprisingly, the largest DNN architectures lead to the highest runtime. 

The runtime of our MIPs heavily depends on the size of the ``big-M'' variable $L$. In practice, $L$ should be chosen as small as possible to obtain a MIP formulation that is as tight as possible. We initialized $L:=3000$ and tightened this bound further by using interval arithmetic (see, e.g., \cited{tjeng2018evaluating}). Recently, \cited{singh2018fast} proposed a novel technique for tightening such MIP formulations. Evaluating this in more detail is subject to future work.
\begin{table}[t!]
	\resizebox{\columnwidth}{!}{
		\begin{tabular}{cccc}
			\toprule
			Domain & $\varnothing$ MIP Runtime & $\varnothing$ Iteration Runtime & $\varnothing$ Auction Runtime  \\
			\midrule
			GSVM & 15.90 sec & 30.51 sec & 44 min \\
			LSVM & 39.75 sec & 51.69 sec & 65 min \\
			MRVM & 3.67 sec &  26.75 sec & 457 min\\
			\bottomrule
	\end{tabular}}
	\caption{Average runtime results of the selected DNN winner models in different domains. All values are averaged over $100$ (GSVM and LSVM) and $50$ (MRVM) auction instances.}
	\label{DNNruntime}
\end{table}
\begin{figure}[t!]
        \centering
        \includegraphics[width=1\columnwidth]{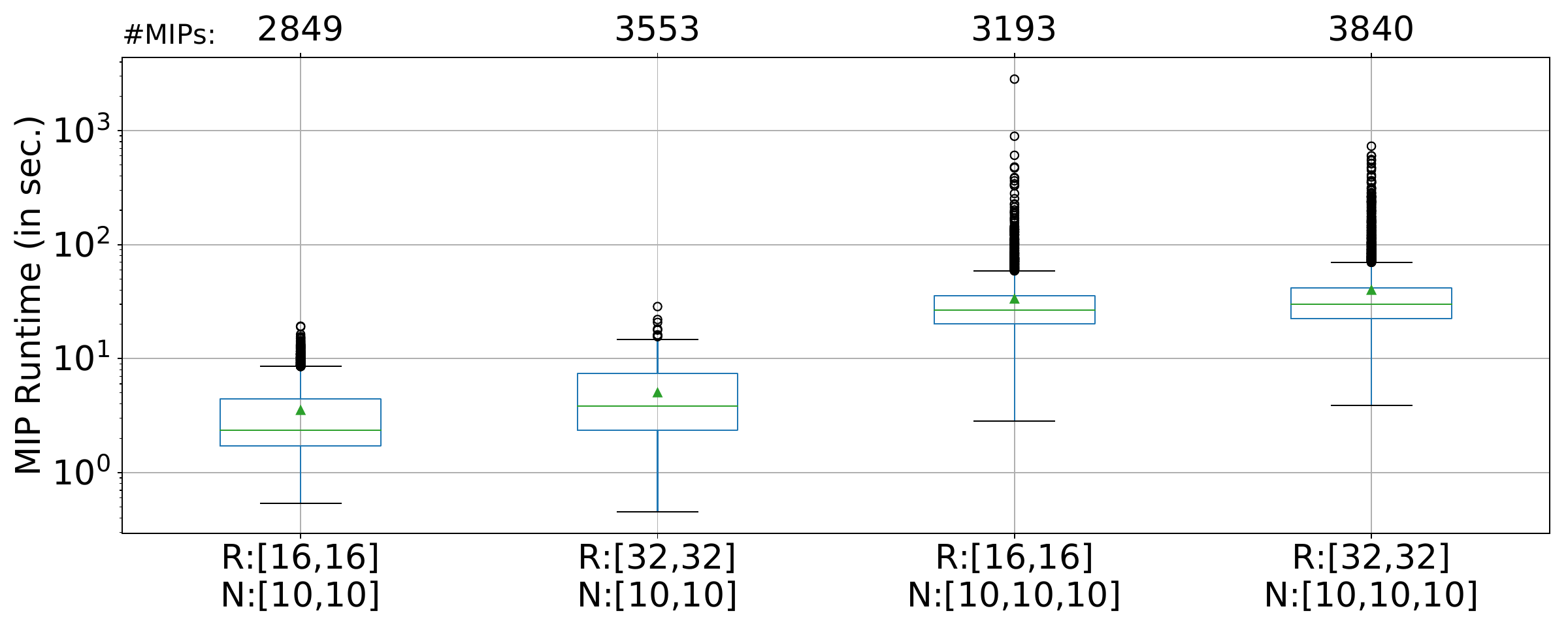}
        \caption{MIP runtimes defined in \eqref{MIP} based on $50$ different LSVM instances. Results are shown for a selection of various DNN architectures and for $c_0:=40, c_e:=50$.}
        \label{fig:miplsvm}
\end{figure}

\section{Conclusion}\label{ConlcusionandFutureWork}
In this paper, we have designed a deep learning-powered ICA. We have compared our approach against prior work using SVRs with quadratic kernels. Our experimental results have shown that our DNN-based approach leads to significantly higher economic efficiency in complex auction domains and scales well to large domains.

On a technical level, our main contribution was to reformulate the DNN-based WDP into a MIP. The main insight to achieve this was to use ReLU activation functions, which can be re-written as multiple linear constraints. From an experimental point of view, we were pleasantly surprised to see that even DNNs with a small number of layers and nodes and with a small number of training samples (i.e., bids) were able to achieve high prediction performance and ultimately high economic efficiency in the overall auction mechanism.

Future work should investigate the trade-off between larger DNN architectures and the resulting MIP runtime, to further increase efficiency. 

\section{Follow-Up Work After Publication}\label{Updates}
After the publication of the present paper, there has been a stream of follow-up papers, which have (a) further improved the ML capability of the mechanism, and (b) applied DNN-based preference elicitation to other combinatorial assignment domains.

\cited{weissteiner2022fourier} designed a \emph{Fourier analysis-based ICA} that leverages different notions of Fourier sparsity. This facilitates the intricate learning task in ICAs where only a few bids (i.e., training samples) can be elicited.

\cited{weissteiner2022monotone} designed \emph{monotone-value neural networks (MVNNs)}, a novel class of DNNs, which by design incorporate \emph{free disposal}. They experimentally showed that MVNNs lead to better generalization performance (especially in settings with only a few bids) and also to higher efficiency (when integrated into an ICA).

\cited{weissteiner2023bayesian} implemented uncertainty-based exploration for ICAs using a novel uncertainty quantification method for DNNs \cite{heiss2022nomu}. With their \emph{Bayesian optimization-based combinatorial assignment (BOCA)} mechanism they could even further increase efficiency.

\cited{soumalias2023machine} applied DNN-based preference elicitation in a combinatorial assignment domain different from that of combinatorial auctions. Concretely, \cited{soumalias2023machine} designed an MVNN-based preference elicitation mechanism for course allocation.

\section*{Acknowledgments}
We thank Gianluca Brero, Nils Olberg, and Stefania Ionescu for insightful discussions and the anonymous reviewers for helpful comments. This paper is part of a project that has received funding from the European Research Council (ERC) under the European Union's Horizon 2020 research and innovation programme (Grant agreement No. 805542).

\begin{table*}
	\centering
	\resizebox{2\columnwidth}{!}{
		\setlength\tabcolsep{4pt}
		\begin{tabular}{llrrrrrrrrrrrrrr}
			\toprule
			\multicolumn{2}{c}{\textbf{SATS}}&   \multicolumn{8}{c}{\textbf{DNN}} & \multicolumn{4}{c}{\textbf{MIP}} &  \multicolumn{2}{c}{\textbf{PVM}}\\
			\cmidrule(l{2pt}r{2pt}){1-2}
			\cmidrule(l{2pt}r{2pt}){3-10}
			\cmidrule(l{2pt}r{2pt}){11-14}
			\cmidrule(l{2pt}r{2pt}){15-16}
			\multicolumn{1}{c}{Domain} &  \multicolumn{1}{c}{Bidder} &    \multicolumn{1}{c}{Epochs}    &  \multicolumn{1}{c}{Batch Size} &   \multicolumn{1}{c}{L1\&L2 Regularization}     &  \multicolumn{1}{c}{Learning Rate} & \multicolumn{1}{c}{Architecture} & \multicolumn{1}{c}{Dropout} & \multicolumn{1}{c}{Dropout Rate} & \multicolumn{1}{c}{MinMaxScaler\footnotemark} & \multicolumn{1}{c}{Bounds Tightening} & \multicolumn{1}{c}{$L$} & \multicolumn{1}{c}{Time Limit (sec)} & \multicolumn{1}{c}{Relative Gap} & \multicolumn{1}{c}{$c_0$} & \multicolumn{1}{c}{$c_e$}\\
			\cmidrule(l{2pt}r{2pt}){1-2}
			\cmidrule(l{2pt}r{2pt}){3-10}
			\cmidrule(l{2pt}r{2pt}){11-14}
			\cmidrule(l{2pt}r{2pt}){15-16}
			GSVM  & Regional  &512 & 32& 0.00001& 0.01&[32,32] &True &0.05 &False &IA &3000 &3600 & 0.0001 &30 &50 \\
			& National  &512 & 32& 0.00001& 0.01&[10,10] &True &0.05 &False &IA &3000 &3600 & 0.0001 &30 &50 \\
			\midrule
			LSVM  & Regional  &512 & 32& 0.00001& 0.01&[32,32] &True &0.05 &False &IA &3000 &3600 & 0.0001 &40 &50 \\
			& National  &512 & 32& 0.00001& 0.01&[10,10,10] &True &0.05 &False &IA &3000 &3600 & 0.0001 &40 &50 \\
			\midrule
			MRVM  & Local  &300 & 32& 0.00001& 0.01&[16,16] &True &0.05 &[0, 500] &IA &3000 &3600 & 0.0001 &10 &100 \\
			& Regional  &300 & 32& 0.00001& 0.01&[16,16] &True &0.05 &[0, 500] &IA &3000 &3600 & 0.0001 &10 &100 \\
			& National  &300 & 32& 0.00001& 0.01&[16,16] &True &0.05 &[0, 500] &IA &3000 &3600 & 0.0001 &10 &100 \\
			\bottomrule
	\end{tabular}}
	\vskip -0.2cm
	\caption{Detailed configurations of all DNN winner models for all SATS domains.}
	\label{tab:winner_configs}
\end{table*}
\bibliography{AAAI-WeissteinerJ.6063}
\bibliographystyle{aaai20}
\appendix
\section*{Appendix}
\section{Detailed Winner Configurations}\label{sec:appendix:DNNhyperparameters}
In Table~\ref{tab:winner_configs}, we provide the detailed configurations (i.e., DNN, MIP, and PVM parameters) of the DNN winner models (i.e., DNN-ICA) from Table~\ref{EFFGSVMTEST}, Table~\ref{EFFLSVMTEST} and Table~\ref{EFFMRVMTEST}. Other parameters not listed in Table~\ref{tab:winner_configs} were set to their default values.
\footnotetext{We use \textsc{Scikit Learn's} \textit{MinMaxScaler} to simultaneously scale bidders' value reports (i.e., $v_i(x^{(k)})$) in the generated (initial) training sets $\cup_{i=1}^n B_i^0$ to an interval $[0,u]$ with $u>0$.}
\end{document}